\documentclass{ocp-jca}
\pdfoutput=1

\usepackage[article,utf8x]{jrouquie}
\usepackage{nicefrac}

\usepackage{xspace}
\newcommand{\reglenum}[1]{\ensuremath{\mathfrak{#1}}}
\newcommand{\rn}{\reglenum}
\newcommand{\rnnl}{$\rn{9}_\ell$\xspace}  
\newcommand{\rnnh}{$\rn{9}_h$\xspace}
\newcommand{\rnchl}{$\rn{58}_\ell$\xspace}  
\newcommand{\rnchh}{$\rn{58}_h$\xspace}
\newcommand{\rncdl}{$\rn{110}_\ell$\xspace}  
\newcommand{\rncdh}{$\rn{110}_h$\xspace}
\newcommand{\rncvsl}{$\rn{126}_\ell$\xspace}  
\newcommand{\rncvsh}{$\rn{126}_h$\xspace}

\begin{document}
\title[Coalescing Cellular Automata]{Coalescing Cellular Automata\\
Synchronizing CA by Common Random Source and Varying Asynchronicity}
\hypersetup{pdftitle={Coalescing Cellular Automata -- Synchronizing CA by Common Random Source and Varying Asynchronicity}}

\author{Rouquier, Morvan}
{Jean-Baptiste Rouquier\instref{univ-lyon}\instref{ensl}\instref{ixxi}
\and Michel Morvan\instref{univ-lyon}\instref{ensl}\instref{ixxi}\instref{ehess}\instref{sfi}}
\hypersetup{pdfauthor={Jean-Baptiste Rouquier and Michel Morvan}}

\institute{univ-lyon}{Universit\'e de Lyon}
\institute{ensl}{ENS Lyon, LIP, 46 allée d'Italie, F-69364 LYON cedex 03, France}
\institute{ixxi}{Institut des Syst\`emes Complexes Rh\^one-Alpes (IXXI)}
\institute{ehess}{EHESS}
\institute{sfi}{Santa Fe Institute}

\date{}


\email{\{jean-baptiste.rouquier, michel.morvan\}@ens-lyon.fr}
\maketitle

\begin{abstract}

We say that a Cellular Automata (CA) is coalescing when its execution on two
distinct (random) initial configurations in the same asynchronous mode
(the same cells are updated in each configuration at each time step)
makes both configurations become identical after a reasonable time.
We prove coalescence for two elementary rules, non coalescence for two other, and show
that there exists infinitely many coalescing CA.
We then conduct an experimental study on all elementary CA and show that some rules exhibit a phase
transition, which belongs to the universality class of directed percolation.
\end{abstract}

\hypersetup{pdfkeywords={synchronization, coalescence, asynchronism, asynchronous cellular automata, directed percolation, coupling, robustness, elementary cellular automata, stochastic process, discrete dynamical system, phase transition}}

\section{Introduction}
The \emph{coalescence} phenomenon, as we call it, has been observed for the
first time by Fates~\cite{async_eca_density_robust}, in the context of
asynchronous cellular automata. Coalescing CA exhibit the
following behavior: starting from two different initial random configurations
and running the same updating sequence
(the same cells are updated at each time step in both configurations),
the configurations quickly become identical, i.e. both configurations
not only reach the same attractor, but they also both synchronize their
orbits. This of course appears in trivial situations, for example if the CA
converges to a single fixed point, but in~\cite{async_eca_density_robust} it is also observed it in a case
where the asymptotic orbit is absolutely non trivial.

The goal of this paper is to explore this rather strange emergent phenomenon in
which the asymptotic behavior seems to be only related to the (random) sequence of
update of the cells and not to the initial configuration.
This work shows that, in some cases, the randomness used during evolution is as important as the
one used during initialization: this stochastic dynamic, with high entropy, is
perfectly insensitive to initial condition. There is thus no chaos here.

The results presented here are of two kinds. First, we prove the existence of
infinitely many different (we precise this notion) non trivial coalescent CA.
We also prove the existence of non trivial non coalescent CA.
Secondly, we study by simulation the behavior of all elementary CA (the ECA are
the CA with 1 dimension, two states, and two nearest neighbors) with regards to
this coalescence property, in an asynchronous context in which at each step each
cell has a fixed probability $\alpha$ to be updated. We show that over the 88
different ECA, six situations occur: a/ 37 ECA never coalesce; b/ 21 always
coalesce in a trivial way (they converge to a unique fixed point); c/ 5 always
coalesce on non trivial orbits; d/ 15 combine a/, b/ and c/ depending on
$\alpha$;
e/ 7 enter either full agreement (coalescence) or full disagreement;
last, f/
3 combine e/ with either a/ or c/.

We also study the transition between non coalescence and
coalescence when $\alpha$ varies for the ECA that combine a/ and c/:
there is a phase transition belonging to the universality class of directed percolation.
We thus get a new model of directed percolation, with a few variants.

Unlike many directed percolation models, the limit of the sub-critical regime is
neither a single absorbing state, nor a set of fixed points, but a non trivially
evolving phase.

\medskip
This paper is an extended version of~\cite{ICCS06}.

The paper is organized as follows. Section~\ref{defs} gives definitions and
notations.
Section~\ref{sec:proof-non-coalesce} proves non coalescence for two CA.
We prove in Section~\ref{formal-proof} that, under certain conditions, CA \rn{6}
and \rn{7} (using the Wolfram's numbering of ECA) are coalescing and show how to
construct from them coalescing CA with arbitrarily many states.
We also prove
that CA \rn{15} and \rn{170} either coalesce or enter total
disagreement, each case occurring with probability $\frac12$.
In Section~\ref{experimental-study}, we describe the exhaustive simulation study
of all ECA and then check the directed percolation hypothesis.
Moreover, we observe that some CA exhibit two phase
transitions: one for small $\alpha$ and one for high $\alpha$.

\section{Definitions and notations}
\label{defs}
\begin{definition}
  An
  \emph{asynchronous finite CA} is
  a tuple $(Q,\,d,\,V,\,\delta,\,n,\,\mu)$ where
  \begin{itemize}
  \item $Q$ is the set of \emph{states};
  \item $d\in\N^*$ is the \emph{dimension};
  \item $V=\{v_1,\,\dots,\,v_{|V|}\}$, the \emph{neighborhood},
    is a finite set of vectors in $\Z^d$;
  \item $\delta:Q^{|V|}\to Q$ is the \emph{transition rule};
  \item $n\in\N^*$ is the \emph{size};
  \item ${\cal U}:=(\Z/n\Z)^d$ is the \emph{cell space} (with periodic boundary condition);
  \item $\mu$, the \emph{synchronism}, is a probability measure on $\{0,1\}^{\cal U}$.
  \end{itemize}
A \emph{configuration} specifies the state of each cell,
and so is a function
$c:{\cal U}\to Q$.
\end{definition}
If $x\in\{0,1\}^{\cal U}$, let $|x|_1$ be the number of 1 in the coordinates of $x$.

\medskip
The dynamic on such CA is then the following. Let $c_t$ denote the configuration at time $t$,
($c_0$ is the initial configuration).
Let $\set{M_t}{t\in\N}$ be a sequence of independent identically distributed random variables with
distribution $\mu$.
The configuration at time $t+1$ is obtained by
\begin{equation*}
  c_{t+1}(z):=
  \begin{cases}
    c_t(z) & \text{if } M_t(z)=0\\
    \delta\big(c_t(z+v_1),\,\dots,\,c_t(z+v_{|V|})\big) & \text{if } M_t(z)=1\\
  \end{cases}\ .
\end{equation*}
In other words, for each cell $z$, we apply the usual transition rule if
$M_t(z)=1$ and freeze it (keep its state) if $M_t(z)=0$.

\medskip
There are many possibilities to asynchronously run a CA. We use the two most classical ones~(\cite{fates2006}),
 which are defined as follow.

\begin{itemize}
\item {The Partially Asynchronous Dynamic}.
Let $0<\alpha\le 1$.
For each cell, we update it with probability $\alpha$, independently from its
neighbors.
$\mu$ is thus the product measure of Bernoulli distributions:
$\mu(x):=\alpha^{|x|_1}(1-\alpha)^{|x|_0}$ (with $0^0=1$).
The case $\alpha=1$ corresponds to the synchronous dynamic.

\item {The Fully Asynchronous Dynamic}.
At each step, we choose one cell and update it, which defines
$\mu(x)$ as $1/n$ if $|x|_1=1$ and $\mu(x):=0$ otherwise.

This last dynamic can be regarded as the limit, when $\alpha\to0$, of the partially
asynchronous dynamic (however, to simulate $t$ steps of this dynamic with the partially
asynchronous one, we need more than $nt^2$ steps).

It is also equivalent (for the order of updates) to the following.
Consider a continuous time and give a clock to each cell. Each clock measure the
time before the next update: when a clock reaches zero, its owner is updated.
The clock is then reset to a random time (to choose the next update date),
according to an exponential law.
\end{itemize}

The usual way to visualize the dynamics of a cellular automaton, asynchronous or
not, is the space-time diagram.

\begin{definition}[space-time diagram]
  Figure~\ref{fig:space-time50} is an example.
  The \emph{space-time diagram} is obtained by stacking up the successive configurations, i.e.,
  each configuration is plotted horizontally, above the previous one. Time thus goes upwards.
\end{definition}

\begin{figure}[htp]
  \centering
  \includegraphics[width=4cm]{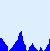}
  \caption{Space-time diagram of \rn{128} (which means ``become $0$ if any of
    your neighbors is $0$'') in the partially asynchronous dynamics.
    $\alpha=0.5$ and the initial configuration has a majority of $1$ (dark cells).
    Space-time diagram of all ECA can be found on
    \url{http://www.rouquier.org/jb/recherche/eca}, with various $\alpha$.}
  \label{fig:space-time50}
\end{figure}

Since we use \emph{pairs} of configurations and study the differences between
both, we will also use two superimposed space-time diagrams, like in Figure~\ref{fig:time-space}.
On such a diagram, we sometimes restrict our attention to cells that agree (ploted light) and cells that don't (plotted dark).
The resulting quotient diagram is called the \emph{agreement/disagreement} space-time diagram.

\bigskip
Let us now introduce the definition of coalescing CA to formalize the observation of~\cite{async_eca_density_robust}.
The principle is to use two initial configurations, and to let them evolve with the
\emph{same} outcome of the random variables $\set{M_t}{t\in\N}$.
In other words, we use two copies of the CA, and at each time step,
we update the same cells in both copies.
This comes down to using the same source of randomness for both copies.

\begin{definition}
  An asynchronous finite CA is \emph{coalescing} if, for any two initial
  configurations,
  applying the same sequence of updates leads both configurations
  to become identical within polynomial expected time (with respect to~$n$).
\end{definition}

Note that it is required for the CA to coalesce quickly enough, where ``quickly'' means in polynomial time.
Section~\ref{sec:proof-non-coalesce} will exhibit a CA that does not coalesce in polynomial time
but nevertheless always coalesce if given enough time.

\smallskip
Any nilpotent CA (converging towards a configuration where all states are
identical) is coalescing if it converges in polynomial time.
But as we will see, and this makes the interest of our study, there exist non
nilpotent coalescing CA, which we call non trivial.

\medskip
The term ``coalescing'' (rather than synchronizing) is used to avoid any confusion with the manipulation of the synchronisms.

In the following, we will often focus on the simplest CA, namely the Elementary CA:
one dimension ($d=1$), $2$ states ($Q=\{0,1\}$), nearest neighbors ($V=\{-1,0,1\}$).
There are $2^8=256$ possible rules, $88$ after symmetry
considerations.
We use the notation introduced by S. Wolfram, numbering the rules from \rn{0} to
\rn{255}. In this notation, a rule $\delta$ is denoted by the number
$2^7\delta(1,1,1)+2^6\delta(1,1,0)+2^5\delta(1,0,1)+2^4\delta(1,0,0)+2^3\delta(0,1,1)+2^2\delta(0,1,0)+2^1\delta(0,0,1)+2^0\delta(0,0,0)$.
For instance, the code for the rule ``minority'' or ``take the least present state of my neighborhood'' (which is \rn{23}) is obtained by reading the last line of Table~\ref{tab:wolfram} as a base 2 number.
We simply write ``\rn{23}'' (with this font) for ``the rule which has code \rn{23}''.
\begin{table}[htp]
  \centering
  \caption{Transition table of \rn{23}}
  \begin{tabular}{rcccccccc}
  neighborhood          & 111 & 110 & 101 & 100 & 011 & 010 & 001 & 000\\
  result of $\delta$    &  0  &  0  &  0  &  1  &  0  &  1  &  1  &  1
\end{tabular}
  \label{tab:wolfram}
\end{table}

Having defined our objects, the next two sections will prove that there exists
both (non trivial) coalescing CA and non coalescing CA.

\section{Analytical Study}
\subsection{Formal proof of non coalescence}
\label{sec:proof-non-coalesce}

Quite obviously, \rn{204}, which means ``identity'', cannot coalesce: disagreeing cells disagree forever.
We would be more interested in a rule that is not coalescing because it does not reach
agreement in polynomial time, but nevertheless always reaches configuration
where all cells agree if given enough time.
In this part, we will prove that this is the case for \rn{60}.
The proof uses specific properties of \emph{affine} rules, which we will introduce first.

\begin{definition}
  A rule is said to be \emph{affine} if it can be expressed as
  $$\delta(q_1,q_2,q_3) = \epsilon+\sum_{i\in I}q_i \mod 2$$
where $I\subseteq\{1,2,3\}$ and $\epsilon\in\{0,1\}.$
\end{definition}
Those rules where introduced in~\cite{mow84}, a detailed study is in~\cite{ccnc97}.

\begin{lemme}
  \label{lemma:affine-quotient}
  For both asynchronous dynamics, affine rules are exactly the rules for which
  the agreement/disagreement space-time diagram is the space-time diagram of a
  cellular automaton.
\end{lemme}
We call the latter CA the \emph{quotient} rule.
\begin{proof}
  First, it is clear that if the rule is affine, we have what we want: if $q_i$
  are the state of one configuration and $q_i'$ are the states of the second,
the agreement/disagreement state is $1$ if $q_i=q_i'$ and $0$ otherwise, which we write $q_i\oplus q_i'$.
In the next configuration, the agreement/disagreement state is
$\delta(q_1,q_2,q_3) \oplus \delta(q_1',q_2',q_3')
= \epsilon+\sum_{i\in I}q_i + \epsilon+\sum_{i\in I}q_i' \mod 2
=\sum_{i\in I} \big(q_i\oplus q_i'\big) \mod 2$,
which is the result of the application of a CA rule $\delta'$ on the states of the previous configuration.
We even see that the quotient rule $\delta$'
is $\delta'(q_1,q_2,q_3) = \sum_{i\in I}q_i \mod 2$, i.e. the linear part of $\delta$.

Reciprocally, let us assume that the agreement/disagreement diagram is still the diagram of a cellular automaton.
If the rule is constant, $I=\emptyset$ is suitable.

Otherwise, there is some configuration $(q_1,q_2,q_3)$ for which $\delta$ changes if a specific $q_i$ changes.
We can choose it to be $q_1$ (the other cases are symmetrical) i.e.
$\delta(q_1,q_2,q_3) \neq \delta(1-q_1,q_2,q_3)$.
Thanks to the initial assumption, this inequality has to be true for all $q_2$ and $q_3$ (indeed, if there would exist
$q_2'$ and $q'_3$ such that $\delta(q_1,\,q_2',\,q_3')=\delta(1-q_1,\,q_2',\,q_3')$,
we would have two pairs of configurations with the same agreement/disagreement pattern at step 0,
but distinct agreement/disagreement pattern at step 1, violating our hypothesis).

The rule can thus be expressed as
$\delta(q_1,q_2,q_3) = \epsilon+q_1+\delta'(q_2,q_3) \mod 2$.
Since this is true for any neighbor which can influence the outcome of $\delta$, the result follows.
\end{proof}

\begin{proposition}
  For both asynchronous dynamics, \rn{60} is not coalescing.
\end{proposition}

\begin{proof}
  \rn{60} is affine, so we can reason on the
  agreement/disagreement diagram. The rule governing this space-time diagram is
  still \rn{60}. This rule means ``XOR between me and my left neighbor''.

Let us distinguish two cases:
  \begin{itemize}
  \item If all cells are in the disagreement state, the total agreement occurs if and only if
    all cells update, which happens with probability $\alpha^n$.
  \item Otherwise, there is a disagreement cell with a left neighbor in the
    agreement state. Whether it updates or not, this cell will stay in the
    disagreement state at the next step.
  \end{itemize}
  The expected time to reach total agreement is thus greater than $\frac1{\alpha^n}$,
  which is asymptotically greater than any polynomial in $n$ for $0<\alpha<1$.
\end{proof}

\begin{remarque}
  If given enough time, \rn{60} always reaches total agreement. Indeed, from any
  configuration, it is possible (though unlikely) to reach agreement: first
  update exactly the agreement cells having a disagreement left neighbor (those
  cell thus go into the disagreement state), and repeat until all cells are in
  the disagreement state. Then update all the cells at once, and total agreement occurs.
  Thus, at each time step, there is a (tiny) non-zero probability to reach total
  agreement.
\end{remarque}

\subsection{Formal proof of coalescence}
\label{formal-proof}
In this section, we will prove that there are infinitely many coalescing CA. For
that, we will prove the coalescence of two particular CA and show how to build an
infinite number of coalescing CA from one of them.

An easy way to do that last point would be to extend a coalescing CA by adding
states that are always mapped to one state of the original CA, regardless of
their neighbors.
However, we consider such a transformation to be artificial since it
leads to a CA that is in some sense identical to the original one.
To avoid this,
we focus on
\emph{state minimal} CA: CA in which any state can be reached (but not necessarily any configuration).
Note that among ECA, only \rn{0} and \rn{255} are not state-minimal.

We will first exhibit two state-minimal coalescing CA
(proposition~\ref{6-et-7-coalescent}); then, using this result, we will deduce the
existence of an infinite number of such CA (theorem~\ref{main-theo});
finally we will describe the coalescent behavior of two others ECA
(proposition~\ref{15-170-semi-coalescent}).
\begin{proposition}
  \label{6-et-7-coalescent}
  \rn{6} and \rn{7} are coalescing for the fully asynchronous dynamic
  when $n$ is odd.
\end{proposition}

\begin{proof}
We call \emph{number of zones} the number of patterns $01$ in a configuration,
which is the number of ``blocks'' of consecutive $1$ (those blocks are the
zones). We first consider only one copy (one configuration).

\smallskip
\begin{table}[htp]
  \centering
  \caption{Transition table of \rn{6}}
  \begin{tabular}{c||c|c|c|c|c|c|c|c}
    Neighbors &1\,1\,1&1\,1\,0&1\,0\,1&1\,0\,0&0\,1\,1&0\,1\,0&0\,0\,1&0\,0\,0\\
    \hline
    New state &  0  &  0  &  0  &  0  &  0  &  1  &  1  &  0\\
  \end{tabular}
  \label{tab:rn6}
\end{table}
Table~\ref{tab:rn6} shows the transition table of \rn{6}.
Since one cell at a time is updated, and since updating the central cell of $101$
or $010$ does not change its state, zones cannot merge,
i.e. the number of zones cannot decrease.

\begin{table}[htp]
  \centering
  \caption{A possible update sequence of \rn{6}.}
$\begin{matrix}
    \cdots\ 0\ \underbrace{0\ 0\ 1}\ \cdots\\
    \cdots\ \underbrace{0\ 0\ 1}\ 1\ \cdots\\
    \cdots\ 0\ \underbrace{1\ 1\ 1}\ \cdots\\
    \cdots\ 0\ 1\ 0\ 1\ \cdots
\end{matrix}$
  \label{tab:update-seq-rn6}
\end{table}
Let us first show that the number of zones actually increase until there are
no patterns $000$ or $111$.
On each pattern $111$, the central cell can be updated (leading to the pattern $101$)
before its neighbors with probability $\frac13$ and with expected time $n$.
On each pattern $0001$, the sequence of Table~\ref{tab:update-seq-rn6} is possible.
It happens without other update of the four cells with probability $1/4^3$ and
with an expected time of $3\,n$.
So, as long as there are patterns $000$ or $111$, the number of zones increases
with an expected time $O(n)$. Since there are $O(n)$ zones, the total expected
time of this increasing phase is $O(n^2)$.

\medskip
The configuration is then regarded as a concatenation of words on $\{0,\,1\}^*$.
Separation between words are chosen to be the middle of each pattern $00$ and $11$,
so we get a sequence of
words that have no consecutive identical letters, each word being at least two
letter long (that is, words of the language ``$(01)^+0?\ |\ (10)^+1?$'').
We now show that borders between these words follow a one way random walk
(towards right) and
meet, in which case a word disappear with positive probability.
The CA evolves therefore towards a configuration with only one word.
Updating the central cell of $100$ does not change its state, so the borders
cannot move towards left more than one cell.
On the other hand, updating the central cell of $001$ or $110$ makes the border move.
One step of this random walk takes an expected time $O(n)$.
The length of a word also follows a (non-biased) random walk, which reaches $1$
after (on average) $O(n^3)$ steps, leading to the pattern $000$ or $111$.
This pattern disappears with a constant non zero probability like in the
increasing phase.
The expected time for $O(n)$ words to disappear is then $O(n^4)$.

\medskip
Since $n$ is odd, the two letters at the ends of the words are the same, i.e.
there is one single pattern $00$ or $11$, still following the biased random walk.
We now consider again the two copies.
Since this pattern changes the phase in the sequence $(01)^+$, it is therefore a
frontier between a region where both configurations agree and a region where
they do not.
The pattern in the other configuration let us come back to the region where the
configurations have coalesced.

\begin{table}[htp]
  \centering
  \caption{Possible update sequences of \rn{6}.}
  \begin{tabular}{c}
    $\begin{matrix}
      \cdots\,1\,0\,1\,0\,0\,1\,0\,1\,\cdots\\[-0.5ex]
      \cdots\,0\,1\,0\,1\,1\,0\,1\,0\,\cdots
    \end{matrix}$
    \\
    a.
  \end{tabular}
  \hfill
  \begin{tabular}{c}
    $\begin{matrix}
      \cdots\,0\,1\,0\,1\,0\,0\,1\,0\,\cdots\\[-0.5ex]
      \cdots\,0\,1\,0\,1\,1\,0\,1\,0\,\cdots
    \end{matrix}$
    \\
    b.
  \end{tabular}
  \hfill
  \begin{tabular}{c}
    $\begin{matrix}
      \cdots\,0\,1\,0\,1\,1\,0\,1\,0\,\cdots\\[-0.5ex]
      \cdots\,0\,1\,0\,0\,1\,0\,1\,0\,\cdots
    \end{matrix}$
    \\
    c.
  \end{tabular}
  \label{tab:update-seqs-rn6}
\end{table}
Let us now study the length of the (single) region of disagreement.
It follows a non biased random walk determined by the moves of both patterns.
When this length reaches $n$, as in Table~\ref{tab:update-seqs-rn6}.a,
the only change happens when the fourth cell is updated, and it decreases the length.
So, the random walk cannot indefinitely stay in state $n$.
On the other hand, when the length reaches $1$, one possibility is Table~\ref{tab:update-seqs-rn6}.b,
where updating the fifth cell leads to coalescence.
The other possibility is Table~\ref{tab:update-seqs-rn6}.c, where updating the fifth then the fourth cell
leads to the former possibility.
In each case, coalescence happens with a constant non zero probability.
One step of this random walk takes an expected time $O(n)$, the total expected
time of the one word step is thus $O(n^3)$
(details on expected time can be found in~\cite{Latin06}).

Therefore, \rn{6} is coalescing.

The proof for \rn{7} is identical, unless that $000$ leads to $010$, which does not
affect the proof (only the increasing phase is faster).
\end{proof}
\begin{remarque}
If $n$ is even, the proof is valid
until there is only one word, at which point
we get a configuration
without $00$ nor $11$, i.e. $(01)^{\nicefrac n2}$ or $(10)^{\nicefrac n2}$. If both copies have
the same parity, it is coalescence, otherwise both copies perfectly disagree (definitively).
Both happen experimentally.
\end{remarque}

\begin{theoreme}
\label{main-theo}
For the fully asynchronous dynamic,
there are non trivial state-minimal coalescing cellular automata with an arbitrarily large number of
states, and therefore infinitely many non trivial state-minimal coalescing CA.
\end{theoreme}
\begin{proof}
Let ${\cal A}^2$ be the product of a CA
${\cal A}=(Q,\,d,\,V,\,\delta,\,n,\,\mu)$ by
itself, defined as $(Q^2,\,d,\,V,\,\delta^2,\,n,\,\mu)$
where\\
\centerline{$\delta^2\big((a,b),\,(c,d),\,(e,f)\big):= \big(\delta(a,c,e),\,\delta(b,d,f)\big)$}
Intuitively, ${\cal A}^2$ is the automaton we get by superposing
two configurations of ${\cal A}$ and letting both evolve according to $\delta$,
but with the same $M_t$. If ${\cal A}$ is state-minimal, so is ${\cal A}^2$.

Let ${\cal A}$ be a coalescing CA. Then ${\cal A}^2$ converges in
polynomial expected time towards a configuration of states all in
$\set{(q,q)}{q\in Q}$.
From this point, ${\cal A}^2$ simulates ${\cal A}$ (by a mere projection of $Q^2$
to $Q$) and is therefore coalescing (with an expected time at most
twice as long); and so are $({\cal A}^2)^2$, $\big(({\cal A}^2)^2\big)^2$, etc,
that form an infinite sequence of CA with increasing size.
\end{proof}

\begin{proposition}
\label{15-170-semi-coalescent}
  \rn{15} and \rn{170}, for both asynchronous dynamics, either coalesce
  or end in total disagreement, each case with probability $\frac12$
  (with respect to the outcome of initial configuration and $(M_t)$).
\end{proposition}
\begin{proof}
  \rn{170} (shift) means ``copy your right neighbor''.
  If one runs two configurations in parallel, they agree on a cell
  if and only if
  they agreed on the right neighbor
  before this cell was updated.
  So (as seen in lemma \ref{lemma:affine-quotient}), 
  the agreement/disagreement space-time diagram is the space-time diagram of a CA.
  This CA is still \rn{170}.
  \rn{170} converges in polynomial time towards $0^n$
  (corresponding to coalescence) or $1^n$ (full disagreement) \cite{Latin06}.
  By symmetry, each case has probability $\frac12$.

  \rn{15} means ``take the state opposed to the one of your right neighbor'',
  and the proof is identical (the quotient CA is still \rn{170}).
\end{proof}

\section{Experimental study and phase transition}
\label{experimental-study}
In this section, we will describe experimental results in the context of partially
asynchronous dynamic. We will show that many ECA exhibit coalescence, then make a finer
classification.
Specifically, we will observe that some ECA undergo a phase transition for this property when
$\alpha$ varies. We will experimentally show that this phase
transition belongs always but in one case to the universality class of directed percolation.

\subsection{Classification of CA with respect to coalescence}
\subsubsection{Protocol}
We call \emph{run} the temporal evolution of a CA when all parameters
(rule, size, $\alpha$ and an initial configuration) are chosen.
We stop the run when the CA has coalesced, or when a predefined maximum running
time has been reached.

Let us describe the parameters we used.
\begin{enumerate}
\item Number of cells $n$.
  The main point when choosing $n$ is to check that the results do not
  depend on a particular choice of $n$.
  Some authors (like \cite{Wuensche}) suggest that small is enough
  ($n=30$), others (like \cite{pt-pas-bon}) state the opposite, and we follow the latter.
  A similar problem studied in~\cite{fatesthese} shows a stable behavior for $n\ge200$.
  We set $n=2\,000$ and check that the results do not change for $n=500$.
\item Number of computation steps.
  To measure the asymptotic density, we let the automaton run for $200\,000$
  steps, then measure the density averaged over $10\,000$ steps.
  (A few cases had to be checked with longer times.)

\begin{remarque}
  According to directed percolation theory, near a phase transition, the
  automaton can take an arbitrarily long time before settling down to the
  asymptotic density.
  So, for a few choices of $\alpha$, those parameters are not sufficient to
  measure the true asymptotic density. However, they are big enough to \emph{detect}
  that there is a transition point, and then study more precisely what happens there.
\end{remarque}

\item Update rate $\alpha$. We try to sample the entire set of possible update rates
  (indeed, we will find transitions distributed among this set).
For each of the 88 rules, we do 999 runs: one for
  each value of $\alpha$ ranging from 0.001 to 0.999.

  One might want to average over many runs. To show that we do not need to, we plot the asymptotic density $\rho$ versus $\alpha$.
  The smoothness of the resulting curve (Figures~\ref{fig:dvsa-non-coalesce}
  and~\ref{fig:dvsa-coalesce}) shows that the variance between runs is low.

  The random seed for deciding which cells to update at each step is distinct for each value of $\alpha$.
\item The initial configuration is random and distinct for each value of $\alpha$.
\end{enumerate}

\subsubsection{Results}
We get the following empirical classes of behaviour:
\begin{itemize}
\item[a/] Some CA never coalesce (or take a too long time to be observed):
  \rn{4}, \rn{5}, \rn{12}, \rn{13}, \rn{25},
  \rn{28}, \rn{29}, \rn{33}, \rn{36}, \rn{37}, \rn{41}, \rn{44}, \rn{45},
  \rn{51}, \rn{54}, \rn{60}, \rn{72}, \rn{73}, \rn{76}, \rn{77}, \rn{78},
  \rn{90}, \rn{94}, \rn{104}, \rn{105}, \rn{108}, \rn{122}, \rn{132}, \rn{140},
  \rn{142}, \rn{150}, \rn{156}, \rn{164}, \rn{172}, \rn{200}, \rn{204},
  \rn{232}.

  One can be a bit more precise by plotting the asymptotic density $\rho$ versus $\alpha$.
  No coalescence means that this curve is always above the $x$ axis.
  We most of the time get a
  noisy curve, close to a horizontal line usually around 0.5
  (Figure~\ref{fig:dvsa-non-coalesce}.a).  This noise shows the precision of the
  measure. It is to be compared to Figure~\ref{fig:dvsa-non-coalesce}.b (\rn{204}, which means ``identity'')
  showing the variance in the initial configuration: for this rule, the
  asymptotic density is the same as the initial density.
  The fact that \rn{150} (Figure~\ref{fig:dvsa-non-coalesce}.c) is ``mixing''
  and has a density converging to $0.5$ seems hard to prove.
  Figures~\ref{fig:dvsa-non-coalesce}.d and~\ref{fig:dvsa-non-coalesce}.e are the other types of plots that arise.

  \medskip
  The case of \rn{142} (Figure~\ref{fig:dvsa-non-coalesce}.f) can be solved
  analytically. In this rule, a cell can change its state only if it is in the
  configuration 001 or 110.  The neighboring cells cannot be in one of these
  configurations and cannot change their states.  Thus, the number of $1^*$
  blocks (or equivalently, the number of occurrences of the word $01$) is
  conserved.  So the automaton can coalesce only if those numbers are equal in
  both initial configurations. The probability for this to occurs tends to 0 as
  $n$ grows.

  Experimentally, when it does not coalesce, there is no asymptotic density for
  this rule: the density keeps evolving. The rule still has correlation between
  both configurations, with larges zones of agreement and large zones of
  disagreement.

  \begin{figure}[p]
    \centering
    \begin{tabular}{c}\includegraphics[width=0.45\textwidth]{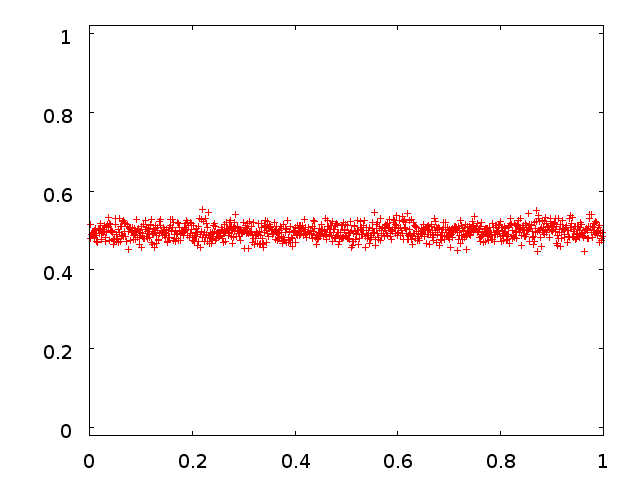}\\a.  \rn{28}\end{tabular}
    \begin{tabular}{c}\includegraphics[width=0.45\textwidth]{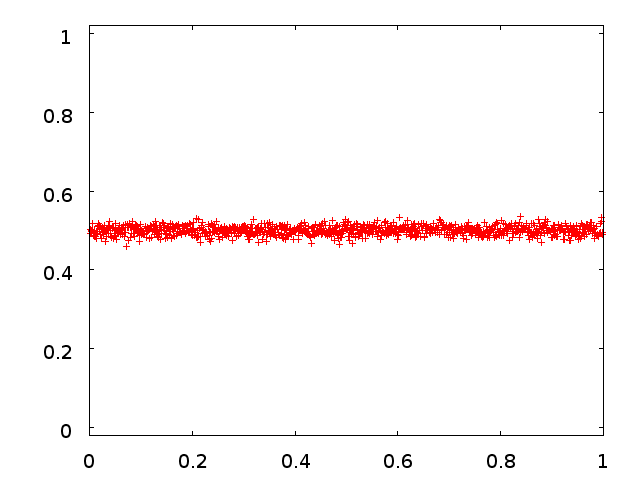}\\b. \rn{204}\end{tabular}
    \begin{tabular}{c}\includegraphics[width=0.45\textwidth]{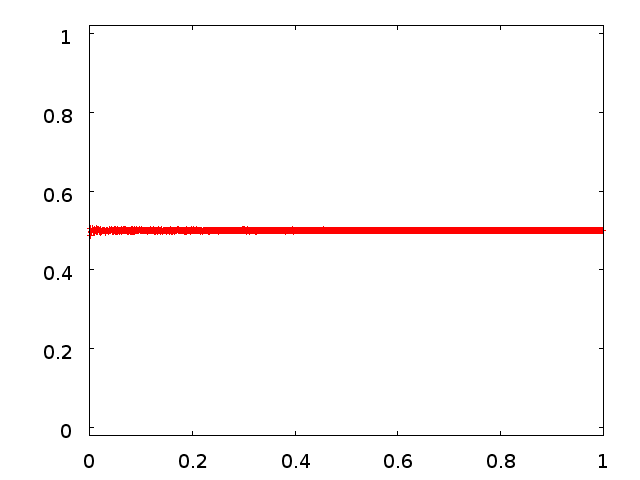}\\c. \rn{150}\end{tabular}
    \begin{tabular}{c}\includegraphics[width=0.45\textwidth]{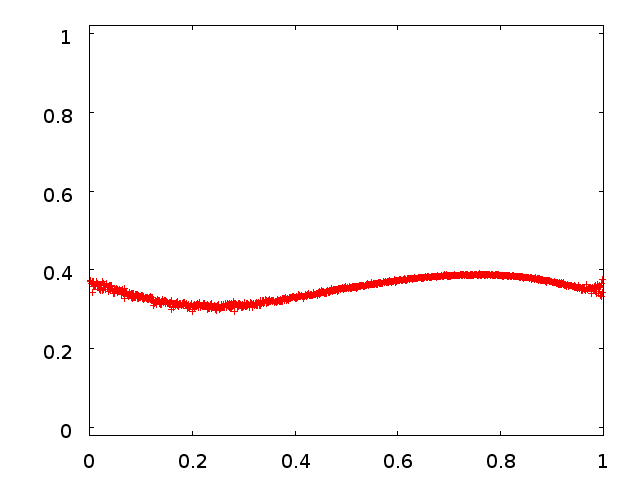}\\d.  \rn{33}\end{tabular}
    \begin{tabular}{c}\includegraphics[width=0.45\textwidth]{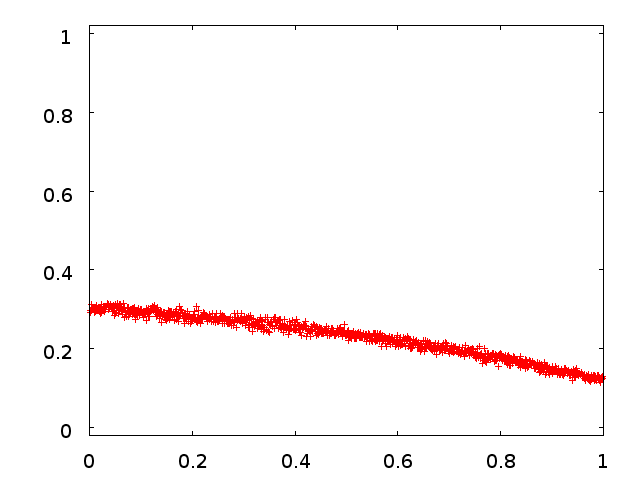}\\e.  \rn{36}\end{tabular}
    \begin{tabular}{c}\includegraphics[width=0.45\textwidth]{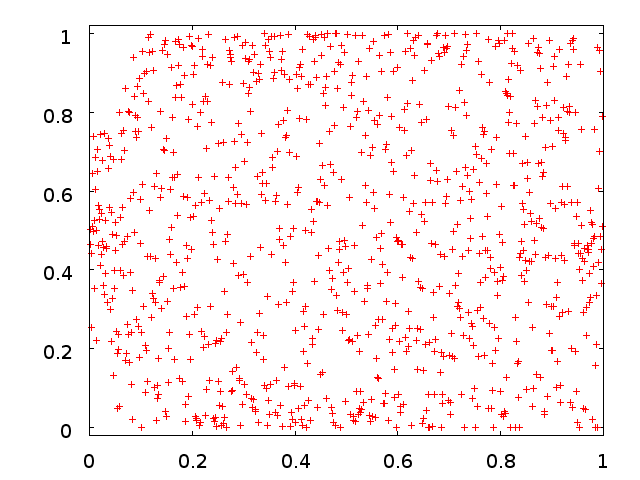}\\f. \rn{142}\end{tabular}
    \caption{Asymptotic density of disagreement cells versus $\alpha$ for some non coalescing rules.
      The scale of both axis is $[0;1]$.
      The first case is typical.}
    \label{fig:dvsa-non-coalesce}
  \end{figure}

\item[b-c/]
  Some CA coalesce rapidly.
\begin{itemize}
\item[b/] The trivial way to do this is to converge to a unique fixed point. One can
  consider the two copies independently, and wait for them to reach the fixed point,
  the CA has then coalesced. This is the case for
  \rn{0}, \rn{2}, \rn{8}, \rn{10}, \rn{24}, \rn{32}, \rn{34}, \rn{38},
  \rn{40}, \rn{42}, \rn{56}, \rn{74}, \rn{128}, \rn{130}, \rn{134}, \rn{136}, \rn{138},
  \rn{152}, \rn{160}, \rn{162}, \rn{168}.

  For most ``captive'' automata (in the ECA case, captive rules are rules for which 000 leads to 0 and 111 leads to 1),
  it has been proven in~\cite{Latin06} whether the rule converges to $0^n$ or not.
  Those rules that provably tends to $0^n$ independently, and thus are coalescing, are:
  \rn{128}, \rn{130}, \rn{136}, \rn{138}, \rn{152}, \rn{160}, \rn{162}, \rn{168}.
\item[c/] The non trivial rules are
  \rn{3}, \rn{19}, \rn{35}, \rn{46}, \rn{154}.
  Rule \rn{19} converge especially slowly for small $\alpha$.
\end{itemize}
\item[d/] Some CA combine the previous behaviors, depending on $\alpha$ (Figure~\ref{fig:dvsa-coalesce}).
  \rn{6}, \rn{18}, \rn{26}, \rn{106}, \rn{146} combine a/ and b/;
  \rn{50} combines b/ and c/;
  \rn{1}, \rn{9}, \rn{11}, \rn{27}, \rn{57}, \rn{62}, \rn{110}, \rn{126} combine a/ and c/ (see Figure~\ref{fig:time-space});
  \rn{58} combines a/, b/ and c/.

  \begin{figure}[p]
    \centering
    \begin{tabular}{c}\includegraphics[width=0.45\textwidth]{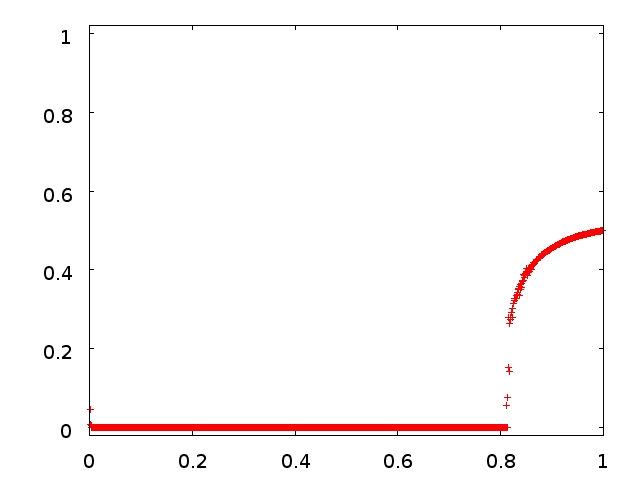}\\a. \rn{106}\end{tabular}
    \begin{tabular}{c}\includegraphics[width=0.45\textwidth]{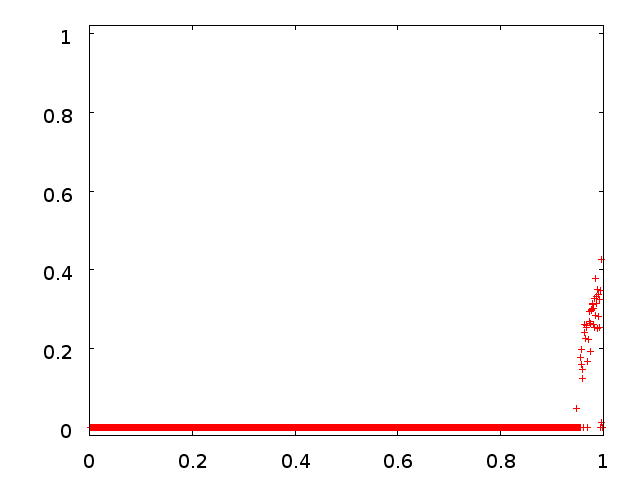}\\b.  \rn{11}\end{tabular}
    \begin{tabular}{c}\includegraphics[width=0.45\textwidth]{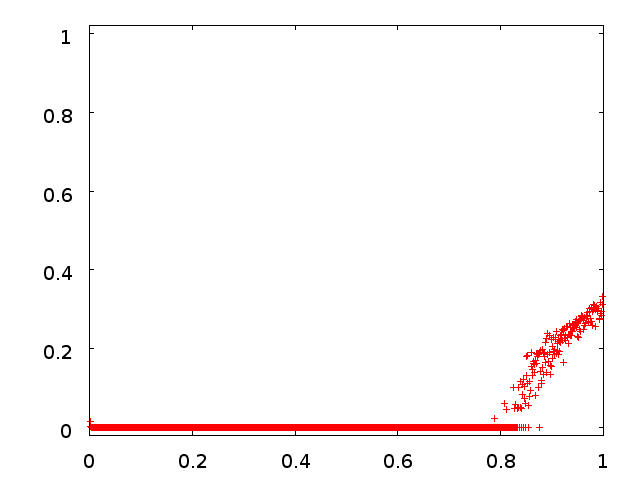}\\c.  \rn{27}\end{tabular}
    \begin{tabular}{c}\includegraphics[width=0.45\textwidth]{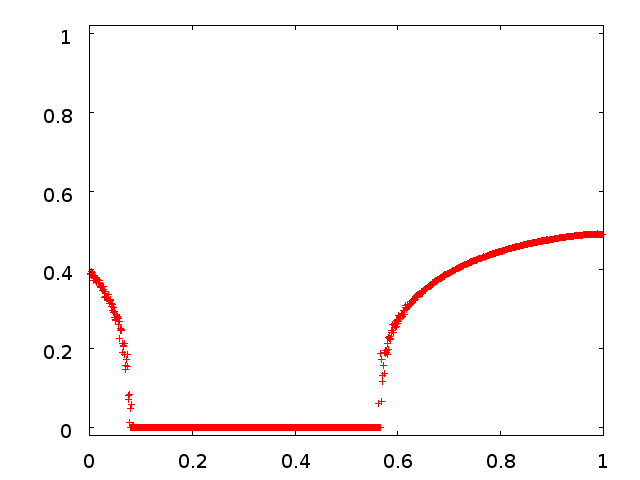}\\d. \rn{110}\end{tabular}
    \begin{tabular}{c}\includegraphics[width=0.45\textwidth]{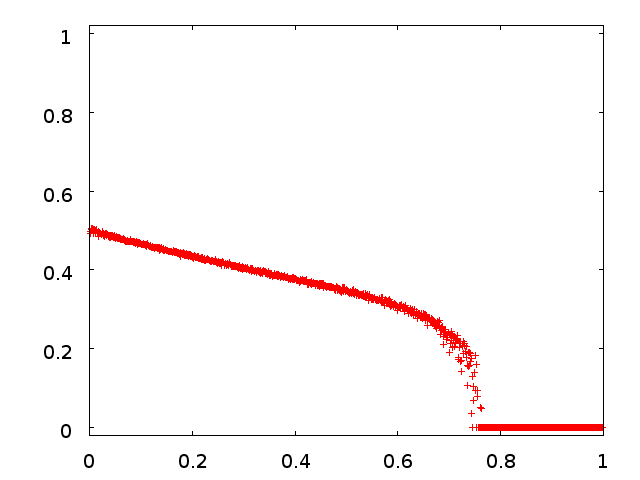}\\e.  \rn{57}\end{tabular}
    \begin{tabular}{c}\includegraphics[width=0.45\textwidth]{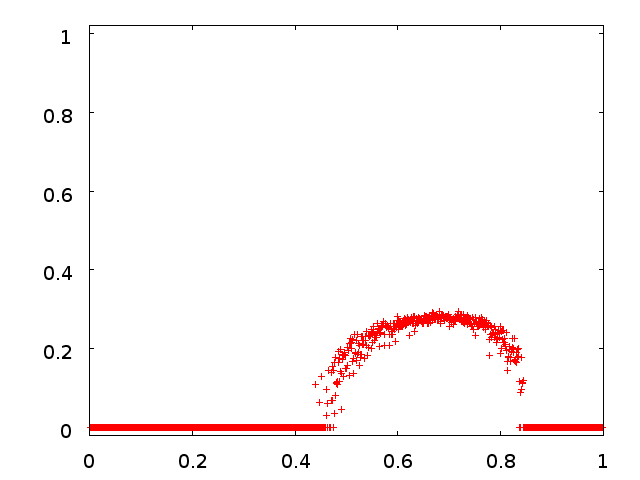}\\f.  \rn{58}\end{tabular}
    \caption{Same as Figure~\ref{fig:dvsa-non-coalesce}, but for some rules which coalesce or not, depending on $\alpha$.
      The first case is the most common.
      Those rules, undergoing a transition, are studied in greater detail in Section~\ref{directed-perco-in-our-model}.}
    \label{fig:dvsa-coalesce}
  \end{figure}
\item[e/] Some CA end in either full agreement between configurations (coalescence)
  or full disagreement, depending on the outcome of $(M_t)$ and the initial configuration:
  \rn{14}, \rn{15}, \rn{23}, \rn{43}, \rn{170}, \rn{178}, \rn{184}.

  These CA take a longer time to reach this asymptotic behavior (this time experimentally seems to be
  $\Theta(n^2)$), so the asymptotic density versus $\alpha$ plots are not two
  horizontal lines, one at $\rho=0$ and the other at $\rho=1$ (Figure~\ref{fig:dvsa-agree-or-disagree}).  They take an even
  longer time for some specific values of $\alpha$
  (thus the ``bow tie'' plot: the closer $\alpha$ is to this specific value,
  the longer the density stays close to 0.5):
  low $\alpha$ (Figure~\ref{fig:dvsa-agree-or-disagree}.a),
  $\alpha=0.5$ (Figure~\ref{fig:dvsa-agree-or-disagree}.b)
  or elsewhere (Figure~\ref{fig:dvsa-agree-or-disagree}.c).
  \begin{figure}[p]
    \centering
    \begin{tabular}{c}\includegraphics[width=0.45\textwidth]{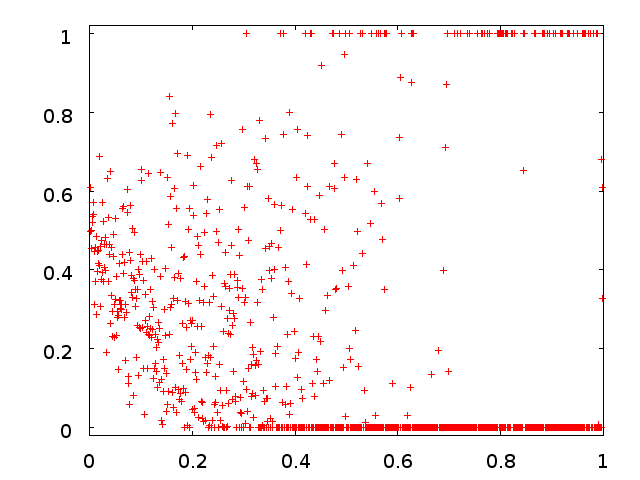}\\a. \rn{184}\end{tabular}
    \begin{tabular}{c}\includegraphics[width=0.45\textwidth]{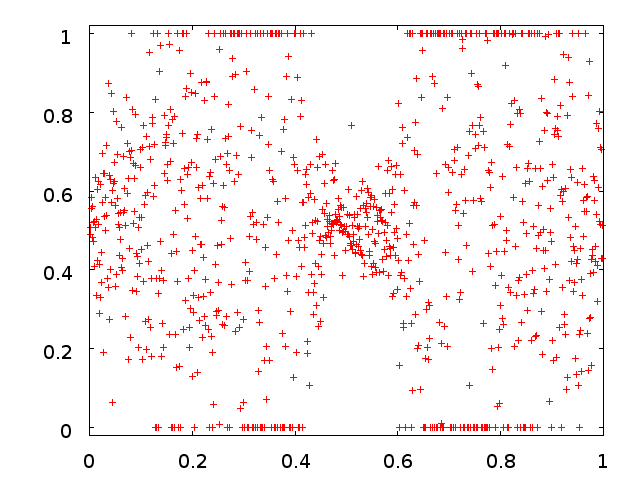}\\b. \rn{178}\end{tabular}
    \begin{tabular}{c}\includegraphics[width=0.45\textwidth]{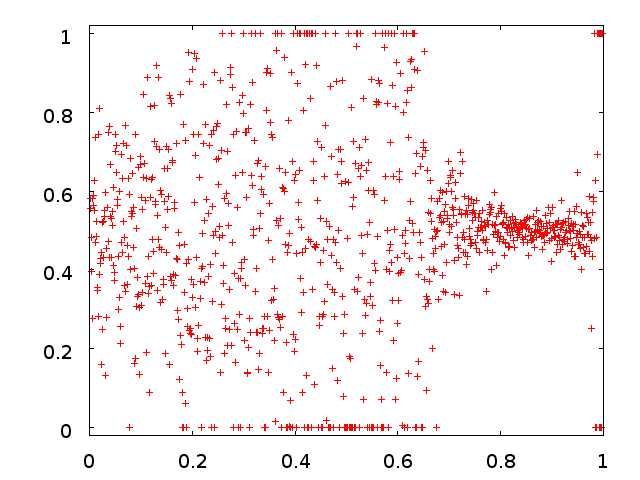}\\c.  \rn{43}\end{tabular}
    \caption{Same as Figure~\ref{fig:dvsa-non-coalesce}, but for some rules ending either in full agreement or full disagreement.}
    \label{fig:dvsa-agree-or-disagree}
  \end{figure}
\item[f/]
  \rn{22} and \rn{30} combine  the previous point (for small $\alpha$) with a/,
  \rn{7}              combines it                 (for small $\alpha$) with c/.
  Note that \rn{30} also shows a varying asymptotic density as it goes from a/ to e/, but due to the difficulties mentionned in e/ this asymptotic density has not been measured.
\end{itemize}

Let us now study the phase transition ``coalescence or not'' when $\alpha$ changes (point d/),
that is,
\rn{1}, \rn{9}, \rn{11}, \rn{27}, \rn{57}, \rn{58}, \rn{62}, \rn{110}, \rn{126}
and \rn{6}, \rn{18}, \rn{26}, \rn{106}, \rn{146}.
Some rules (\rn{9}, \rn{58}, \rn{110}, \rn{126}) show two phase transitions,
one for $\ell$ow $\alpha$, denoted by a subscript $\ell$ like in \rnnl, one for
$h$igh $\alpha$, denoted by a subscript $h$ like in \rnnh.

The next subsection (\ref{Phase transition and directed percolation}) recalls the useful background before studying this phase transition:
it briefly recalls what a phase transition is,
it presents a conjecture implying that our transition belongs to the directed percolation universality class,
and describes this directed percolation class.

\subsection{Phase transition and directed percolation}
\label{Phase transition and directed percolation}
\subsubsection{Phase transition}
A phase transition is an abrupt change in macroscopic properties of a system
with only a small change of control a parameter, say $T$, around a critical value $T_c$.
This paper is concerned only with second order phase transitions,
or continuous phase transitions, which can be characterized by \emph{critical exponents}.
If one let the parameter $T$ vary near the phase transition (occuring at $T=T_c$), all other variables being fixed, a measurable quantity $C$ has a power law behaviour
$C\propto |T-T_c|^\beta$ at least on one side of $T_c$.
Several exponents are defined, depending on the quantity measured.

Remarkably, many systems with no a priori relation turn out to have the same critical exponents.
A \emph{universality class} is defined as all the systems having the same set of critical exponents.

\subsubsection{A conjecture on damage spreading}
Chaos theory deals with the sensitivity to initial condition of \emph{deterministic} systems.
To also study the influence of small perturbations on stochastic systems, \cite{KaufgmanDamageSpreading} introduced \emph{damage spreading}.
In this model, two copies of a stochastic model are run in parallel with the same source of random bits,
starting from different initial configurations (often they are set to differ in exactly one site).

One measures the temporal evolution of the proportion of differing sites, called the Hamming distance.
If this goes to zero, i.e. if both copies become identical, the initial ``damage'' has ``healed'', otherwise the damage is said to spread.

There is a conjecture by~\cite{grassberger-damagespread-in-dp} stating that,
if a transition occurs between healing and spreading in a stochastic spin model,
the universality class of this phase transition is always the same,
namely the one of \emph{directed percolation}, which will be presented in the next paragraph.

There are some conditions for this conjecture:
\begin{enumerate}
\item Only short range interactions in time and space,\label{enum:short}
\item translational invariance\label{enum:translat},
\item non vanishing probability for a site to become healed locally\label{enum:heal},
\item the transition does not coincide with another phase transition\label{enum:other-transition}.
\end{enumerate}

Points~\ref{enum:short} and~\ref{enum:translat} are easily fullfilled for CA.
We will discuss points~\ref{enum:heal} and~\ref{enum:other-transition} in Section~\ref{sec:measure-alpha}.

\subsubsection{The Model of Directed Percolation}
\label{directed-perco}
An more detailed introduction to directed percolation can be found
in~\cite{fates2006} (note that this papers cites a different conjecture of Grassberger than the one we deal with).
A survey of directed percolation is contained in~\cite{hinrichsen-2000-49}, which also covers \emph{damage spreading}.

Isotropic percolation was first defined when studying propagation of a fluid through a porous medium.
It has been mathematically modelled as an infinite square grid where each site has the four nearest sites as neighbors.
Each bond between two neighbors can be open (letting the fluid go through) with probability $p$ or closed with probability $1-p$, independently of all other bonds.

The question is whether the fluid inserted at one point will pass through the
medium, i.e. whether this point is part of an infinite network of sites
connected by open bonds.

Directed percolation appears when one adds gravity to the model,
 i.e. when the fluid is only allowed to travel in one direction (Figure~\ref{fig:perco}).
Static 2D directed percolation can also be seen as a 1D dynamical model where some sites are ``active'' (where active can mean wet, infected, etc.).
An active cell can stay active or die (become inactive), and make its neighbors active.
Depending on the probabilities of these possibilities, active regions spread or disappear.
Cells can only have an influence on the future states of their neighbors, thus the \emph{directed} percolation.

The order parameter measured is the density of active states as a function of $p$ and time, $\rho(p,t)$.
It is zero in one phase and non-zero in the other.
There exists a critical probability $p_c$ which is the limit between two phases.
\begin{itemize}
\item For $p<p_c$, the asymptotic density $\rho(p,\infty)$ is $0$;
\item for $p>p_c$ we have a power law  $\rho(p,\infty)\propto (p-p_c)^\beta$;
\item for $p=p_c$ the density goes to $0$ as $\rho(p_c,t)\propto t^{-\delta}$.
\end{itemize}

\begin{figure}[htp]
  \centering
  \includegraphics[width=0.8\textwidth]{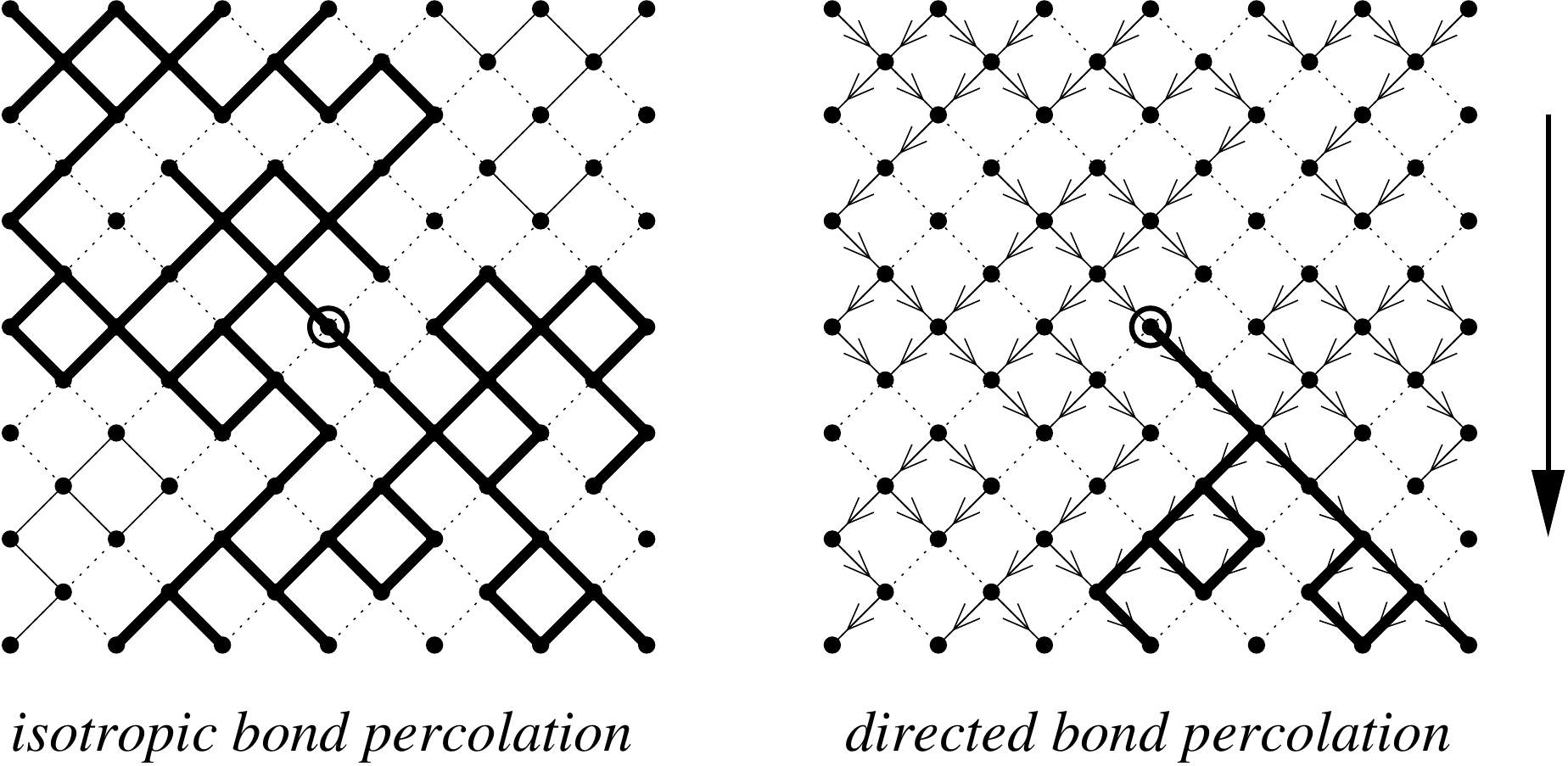}
  \caption{Isotropic (left) and directed (right) percolation. Figure reprinted from~\cite{hinrichsen-2000-49}.}
  \label{fig:perco}
\end{figure}

\subsection{Directed percolation in our model}
\label{directed-perco-in-our-model}

Our active sites are the cells where the configurations disagree.
Density of such sites is written $\rho(\alpha,t)$, or $\rho(\alpha)$ for the asymptotic density.
The pairs of configurations where all cells agree constitute the absorbing set.
Percolation transition (coalescence or not) appears when varying $\alpha$, see Figure~\ref{fig:time-space}.
Note that there is no direct relation between $\alpha$ and $p$.
The aim is thus to identify $\beta$ assuming that
$\rho(\alpha) \propto\,|\alpha-\alpha_c|^{\beta}$ for some $\alpha_c$.
Like many authors, we will
focus on $\beta$ and consider it as sufficient to test directed percolation.

\begin{figure}[htp]
  \centerline{\hbox{
    \includegraphics[width=0.49\linewidth]{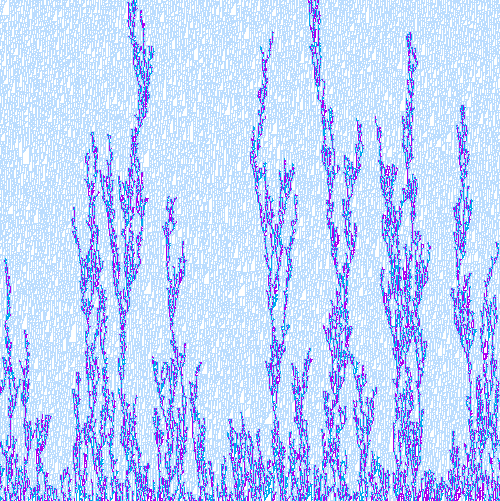}
    \includegraphics[width=0.49\linewidth]{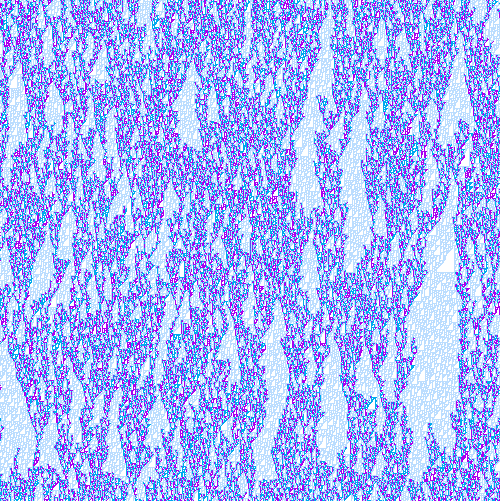}}}
  \caption{%
    Rule \rn{110}, $n=500$.
    Time goes upwards, during $500$ steps.
    Sites where both configurations disagree are dark, coalesced sites are light
    (with light blue standing for state $1$, white for $0$).
    Left: sub-critical phase ($\alpha=0.47<\alpha_c\simeq 0.566$), branches die.
    Right: supercritical phase ($\alpha=0.65>\alpha_c$), active sites spread.}
  \label{fig:time-space}
\end{figure}

\medskip
We deal with finite configurations and can thus be subject to finite size effects.
One effect of particular importance is the following.
Take a rule with a phase transition, there is an update rate $\alpha$ for wich the rule is coalescing,
i.e. reaches total agreement in polynomial time.
In the non coalescing regime, with low probability,
  the outcome of the random bits determining which cells get updated can make
  the CA simulate the coalescing regime for a fixed number of steps.
So, if a CA can coalesce for a given $\alpha$, it can coalesce for any $\alpha$ in $(0;1)$.
The true asymptotic regime is thus always coalescence.

However, the long term behaviour we're interested in,
and the one we can observe with practical simulations,
is after a polynomial time only.
It is a ``long transient'' as opposed to the short transient happening after initialization.
The sytem stays in this ``long transient'' regime for a super-polynomial time, then reaches the true asymptotic regime.
(It can be compared to a living system where one studies reactions to and transient regime after various events in life,
but where the asymptotic regime is always death.)
Note that coalescence was precisely defined as coalescence in polynomial expected time.

\subsubsection{Measure of \texorpdfstring{$\alpha_c$}{critical alpha}}
\label{sec:measure-alpha}
To measure $\beta$, we use the method described in~\cite{hinrichsen-2000-49} that advise to first measure $\alpha_c$ before doing a fit to measure $\beta$, instead of fitting $\alpha$ and $\beta$ at the same time.

To measure $\alpha_c$, the method implies plotting
the density $\rho$ of active sites versus time in logarithmic scale and finding the
$\alpha$ value for which one gets a straight line (for $\alpha<\alpha_c$, the AC
coalesce faster, for $\alpha>\alpha_c$, it has a positive asymptotic $\rho$).
We used random initial configuration with each state equiprobable.
To get readable plots we needed up to $n=10^6$ cells and $10^7$ time steps.
We get (recall that $\alpha_c$ is not universal, it is just used to compute
$\beta$):
  \begin{center}
    \begin{tabular}{|l|c@{\ }c@{\ }c@{\ }c@{\ }c@{\ }c@{\ }c@{\ }c@{\ }c|}
      \hline
      rule           &\rn{1}&\rn{6}&\rnnl&\rnnh&\rn{11}&\rn{18}&\rn{26}&\rn{27}&\rn{57}\\ \hline
      $\alpha_c >...$& 0.101& 0.06 &0.073&0.757& 0.9575&0.7138 &0.4747 & 0.856 & 0.749 \\
      $\alpha_c <...$& 0.103& 0.08 &0.074&0.758& 0.9583&0.7141 &0.4751 & 0.858 & 0.750 \\
      \hline
    \end{tabular}\nopagebreak\par\smallskip
    \begin{tabular}{|l|c@{\ }c@{\ }c@{\ }c@{\ }c@{\ }c@{\ }c@{\ }c@{\ }c|}
      \hline
      rule           &\rnchl&\rnchh&\rn{62}&\rn{106}&\rncdl&\rncdh&\rncvsl&\rncvsh&\rn{146}\\ \hline
      $\alpha_c >...$&0.4745&0.8408& 0.598 & 0.8143 & 0.073&0.566 & 0.101 & 0.720 &0.6750  \\
      $\alpha_c <...$&0.4748&0.8412& 0.599 & 0.8147 & 0.075&0.567 & 0.103 & 0.721 &0.6753  \\
      \hline
    \end{tabular}
  \end{center}
Note that the $\alpha_c$ of \rn{1} and \rncvsl, like \rnnl and \rncdl, are very close, and may be equal.
Also, $\alpha_c$ for \rn{6} and \rn{11} are quite close to 0 and 1.

\medskip
\paragraph{Comparison to a simpler model}
In~\cite{fates2006}, a quite similar problem is studied.
The author chooses the more classical model with only one configuration, but still the
same asynchronous updating.
The  ``1'' is taken as the active state, and $0^n$ as the absorbing dead state.
This article showed experimentally that some rules
(\rn{6}, \rn{18}, \rn{26}, \rn{50}, \rn{58}, \rn{106} and \rn{146})
belong to the directed percolation class: they
converge to $0^n$ for small enough $\alpha<\alpha_c$ for some $\alpha_c$, and
the asymptotic density $d'$ of cells in the state ``1'' near $\alpha_c$ is
$C(\alpha-\alpha_c)^\beta$ for some constant $C$.

Such rules are clearly trivially coalescing for low $\alpha$: both configurations independently converge to $0^n$.
So, if one of these rules is not coalescing for high $\alpha$, it will undergo a phase transition in our model.
Here is a review of those rules:

\begin{itemize}
\item \rn{50} is always coalescing, either trivially or not, and thus isn't of interest in the present work.
  The other rules are not always coalescing.
\item \rn{58} has a different $\alpha_c$ than in our work and is studied in Section~\ref{measure-beta}.
\item \rn{6} also has a different $\alpha_c$. However, near $\alpha_c$, both configurations
  quickly becomes mostly composed of periodic regions repeating the pattern ``01''.
  Depending on the parity, both configurations agree or disagree on the whole overlap of such regions.
  Inside such an overlap, a cell state does not change, whether updated or not.
  So, point~\ref{enum:heal} of Grassberger's conjecture is not fullfilled: a single site cannot always heal.
  Indeed, the density versus time plot is not a power law in the form $t^{-\delta}$, as one would expect for directed percolation.
  See Figure~\ref{fig:rn6}.
\item Finally,
\rn{18}, \rn{26}, \rn{106} and \rn{146} seem to have the \emph{same} $\alpha_c$
than the one measured in~\cite{fates2006}
(at least up to the available precision $\pm3.10^{-4}$).

In our model with two configurations, if we assume both configurations to be
independent\footnote{The present article studies specifically the cases where
  both configurations are equal, but since there is only trivial coalescence, or
  no coalescence for those rules, this assumption is not trivially false.}, then the
density of disagreement cells is $2d'(1-d')$
(as defined in the previous paragraph, $d'$ is the density of ``1'' in the model with only one configuration).
$2d'(1-d') = 2 C(\alpha-\alpha_c)^\beta
(1-C(\alpha-\alpha_c)^\beta) = 2C(\alpha-\alpha_c)^\beta-2C^2(\alpha-\alpha_c)^{2\beta}$.
The second term is negligible (second order) near $\alpha_c$, so (under the
assumption of independence) the rule still belongs to the directed percolation
class.
Note however that this is harder to observe experimentally because of the higher precision needed (precision that we do not currently reach, although we observe lower $β$ for those rules, and greater sensitivity to the sampling interval for $α$).
Anyway, those rules do not meet point~\ref{enum:other-transition} of Grassberger's conjecture, since there is another phase transition happening in the underlying, single configuration, CA.

For both reasons, we leave those rules for future work.
\end{itemize}

\begin{figure}[htp]
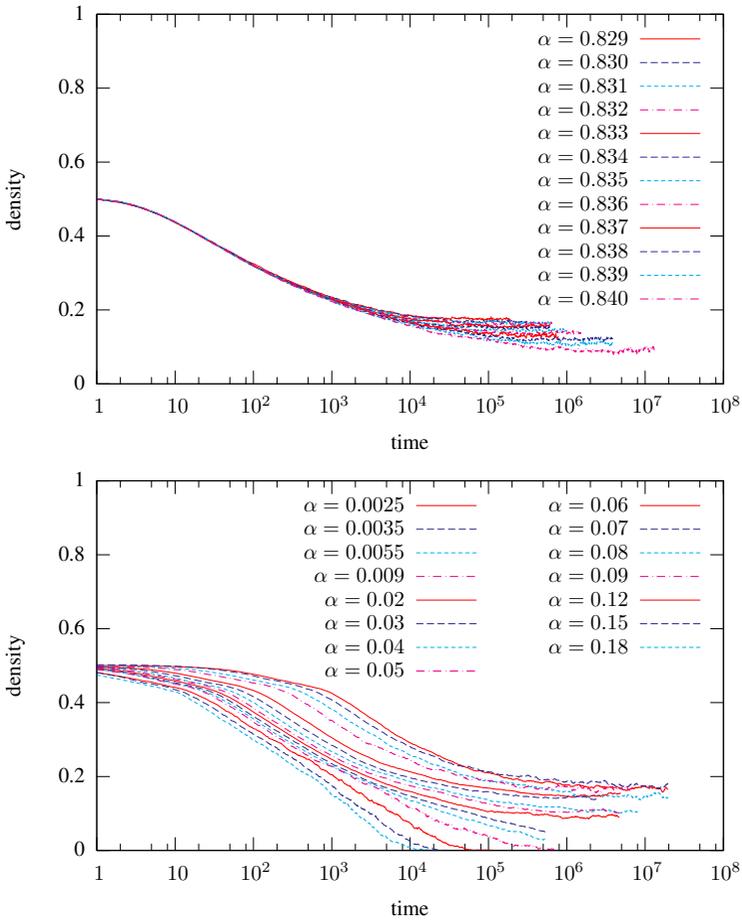

  \centering
  \scalebox{0.8}{\input{figures/rn58.gnuplot_tex}}
    \scalebox{0.8}{\input{figures/rn6.gnuplot_tex}}
  \caption{Density versus time for a rule in the directed percolation class (\rn{58}, top) and for \rn{6} (bottom).
  The asymptotic density goes to $0$ as $α→α_c$ (and is $0$ afterwards, i.e. coalescence happens).}
  \label{fig:rn6}
\end{figure}

\subsubsection{Measure of \texorpdfstring{$\beta$}{beta}}
\label{measure-beta}

Let us now plot $\rho$ versus $\alpha$ near $\alpha_c$
 (Figure~\ref{fig:fit_beta}).
We actually plot $\log\rho$ versus $\log(\alpha-\alpha_c)$ and fit a straight line, the slope of which is an estimator of $\beta$.
It is important to do the fit against $\log\rho$ (and not $\rho$),
so that all errors get the same weight when fitting a line on the log-log plot.
It is also possible to adjust $\alpha_c$ to get a straight line.
This method experimentally has the same computing time/precision ratio as determining $\alpha_c$ beforehand.

\begin{figure}[tbp]
  \centerline{%
    \quad
    \scalebox{0.8}{\input{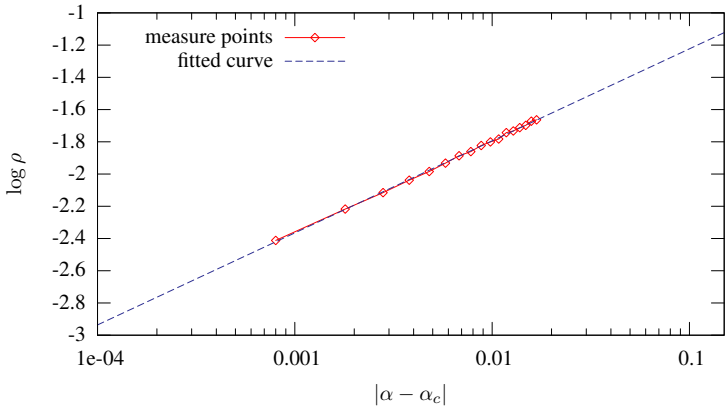}}%
  }
  \caption{Measuring $\beta$ for \rnchh.}
  \label{fig:fit_beta}
\end{figure}

\textbf{Protocol}\quad
The protocol is only semi-automatic.
We choose $n$ between $10\,000$ and $1\,000\,000$ to get a reasonably smooth line on the density versus time plot.
We visually check that the density has reached a steady state, then average the density over at least half a decade.
This yields one measure point. We repeat this process for several values of $α$ near $α_c$.

The fit gives the following ranges, taking into account uncertainty about
$\alpha_c$ and which points to keep for the fit.
Experimental value for $\beta$ measured on other systems is $0.276$.
\begin{center}
  \begin{tabular}{|l|c@{\ }c@{\ }c@{\ }c@{\ }c@{\ }c|}
    \hline
    rule        &\rn{1} & \rnnl & \rnnh &\rn{11}&\rn{27}&\rn{57}\\
    \hline
    $\beta >...$& 0.265 & 0.270 & 0.273 & 0.264& 0.258 & 0.248 \\
    $\beta <...$& 0.279 & 0.295 & 0.283 & 0.326& 0.305 & 0.281 \\
    \hline
  \end{tabular}\nopagebreak\par\smallskip
  \begin{tabular}{|l|c@{\ }c@{\ }c@{\ }c@{\ }c@{\ }c@{\ }c|}
    \hline
    rule        &\rnchl&\rnchh&\rn{62}&\rncdl&\rncdh&\rncvsl&\rncvsh\\
    \hline
    $\beta >...$&0.270 & 0.248& 0.270 &0.270 & 0.271&0.250  & 0.260\\
    $\beta <...$&0.297 & 0.274& 0.281 &0.291 & 0.281&0.276  & 0.276\\
    \hline
  \end{tabular}
\end{center}
As expected, all models remaining at the end of Section~\ref{sec:measure-alpha} seem to belong to the universality class of directed percolation.
We cannot definitely conclude about \rn{11} and \rn{27}, due to higher noise and thus lack of precision.

\section{Acknowledgments}
We would like to thanks Peter Grassberger for useful advice, and Nazim Fates for sharing early results.

Source code is available on \url{cimula.sf.net}.

\bibliographystyle{plain}
\bibliography{biblio}
\end{document}